\documentclass[preprint,12pt]{elsarticle}

\usepackage{amsmath,amssymb,amsfonts}
\usepackage{amsthm}
\usepackage{array}
\usepackage{booktabs}
\usepackage{multirow}
\usepackage{tabularx}
\usepackage{makecell}
\usepackage{graphicx}
\usepackage{flafter}
\usepackage{float}
\usepackage{enumitem}
\usepackage{url}
\usepackage{xurl}
\usepackage[hidelinks]{hyperref}
\hypersetup{
  pdftitle={Do We Still Need Demand Flexibility as Batteries Become Cheaper? A Levelized Cost Perspective},
  pdfauthor={Ruike Lyu and Tengmu Li}
}

\graphicspath{{./figures/}}
\journal{Energy Conversion and Management}
\biboptions{sort&compress}

\newtheorem{proposition}{Proposition}
\theoremstyle{remark}

\begin{document}

\begin{frontmatter}

\title{Do We Still Need Demand Flexibility as Batteries Become Cheaper? A Levelized Cost Perspective}

\author[princeton]{Ruike Lyu}
\author[tju]{Tengmu Li}
\address[princeton]{Andlinger Center for Energy and the Environment and Department of Mechanical and Aerospace Engineering, Princeton University, Princeton, NJ 08544, USA}
\address[tju]{School of Electrical and Information Engineering, Tianjin University, Nankai District, Tianjin 300072, China}

\begin{abstract}
Energy storage and flexible loads both help balance power systems with high shares of variable renewables, but they do so differently. A battery moves electricity from one period to another. A factory or data center instead moves production or computing activity, while meeting demand for its product or service. As batteries become cheaper, it is natural to ask whether demand flexibility is still needed. Many assessments treat that controllable load as already available and focus on enabling or participation costs. For a large load that normally runs near full capacity, however, shifting electricity use can require extra equipment, delay sales, tie up working capital, and disrupt the process. Here, we compare energy storage technologies with industrial and data-center load shifting on a levelized cost basis, measuring the incremental cost of flexible relative to baseline operation. We find that energy storage becomes more expensive per shifted kilowatt-hour as the flexibility period lengthens and assets are cycled less often. Load shifting does not share that rise in capacity cost under a fixed avoided-hour share, but new productive capacity can still be too expensive to justify. Under representative central technology costs, batteries compete well for short-cycle shifting but are more than 80 times more expensive than load shifting of aluminum smelting when providing seasonal flexibility. The answer therefore depends on the timescale, the load, and whether the needed capacity is new or already sunk.
\end{abstract}


\end{frontmatter}

\section*{Nomenclature}
\addcontentsline{toc}{section}{Nomenclature}
\begin{description}[style=multiline,leftmargin=3.5cm,labelwidth=3.2cm]
\item[$T$] Flexibility period, or complete shifting-cycle window (h).
\item[$H$] High-price hours shifted, avoided, or discharged within $T$.
\item[$W$] Price window used to calculate system value.
\item[$N_{\mathrm{cycle}}$] Number of cycles per year, $8760/T$ (yr$^{-1}$).
\item[$\Delta t$] High-price operating time avoided within one period (h).
\item[$e_1$] Energy shifted or supplied during high-price hours in one cycle (kWh).
\item[$E_{\mathrm{shift}}$] Shifted energy within one window; the modifier $\mathrm{yr}$ denotes the corresponding annual quantity (kWh and kWh/yr).
\item[$E_{\mathrm{movable}}$] Maximum energy movable without new capacity (kWh or kWh/yr).
\item[$p_0$] Baseline load capability; $p_{\mathrm{flex}}$ is flexible power available for shifting (kW).
\item[$\delta p$] Added load capability within one period; $\Delta p$ is the corresponding added power used in annualized capacity-cost accounting (kW).
\item[$r$] Actual shifted share; $\rho$ is the avoidable high-price share (dimensionless).

\item[$I$] Original one-time capital expenditure (Chinese yuan, CNY, per corresponding capacity unit).
\item[$y$] Asset or process lifetime (yr).
\item[$r_{\mathrm{int}}$] Rate used for capital recovery; $r_{\mathrm{discount}}$ is the inventory-financing rate (yr$^{-1}$).
\item[$\mathrm{CRF}$] Capital recovery factor (yr$^{-1}$).
\item[$f_{\mathrm{FOM}}$] Fixed O\&M as a fraction of investment; $\mathrm{FOM}_{\mathrm{abs}}$ is the absolute annual cost per unit of capacity.
\item[$K_x$] Equivalent annualized capacity cost for component or capacity basis $x$; the relevant unit is CNY/(capacity unit yr).

\item[$\lambda$] Marginal electricity price; a bar denotes an average over the hours identified by the subscripts (CNY/kWh).
\item[$V_W$] Peak-hour average marginal value in one window; $V_{\mathrm{system}}$ is its average across all complete windows (CNY/kWh shifted).
\item[$c_x$] Unit shifted-energy cost of resource or cost component $x$ (CNY/kWh shifted).
\item[$C_{x,\mathrm{yr}}$] Annual total cost of resource $x$ (CNY/yr).
\item[$\eta_x$] Efficiency or utilization factor of component $x$ (dimensionless).

\item[$E_{\mathrm{bat}}$] Battery energy capacity; $p_{\mathrm{bat}}$ and $E_{\mathrm{charge}}$ are the battery power capacity and charging electricity required to supply $e_1$ (kWh, kW, and kWh).
\item[$k_{\mathrm{ely}}$] Electrolyzer power; $k_{\mathrm{fc}}$ and $e_{\mathrm{h2}}$ are the fuel-cell power and hydrogen energy capacity required to supply $e_1$ (kW, kW, and kWh$_{\mathrm{H2}}$).
\item[$e_{\mathrm{intensity}}$] Electricity consumption per unit of product (kWh/t).
\item[$\lambda_{\mathrm{product}}$] Product price or inventory value (CNY/t).
\item[$\tau_{\mathrm{wc}}$] Inventory cash lock-up duration (yr).
\item[$\mathrm{Inv}(t)$] Physical-product inventory at time $t$ (t).
\item[$\mathrm{Production}$] Production time series; $\mathrm{Demand}$ is the product-demand time series (t per modeled interval).
\item[$\mathrm{Area}(t)$] Warehousing footprint at time $t$ (m$^2$).
\item[$\rho_{\mathrm{material}}$] Material stacking density; $h_{\mathrm{effective}}$ and $\eta_{\mathrm{storage}}$ are effective stacking height and warehouse-space utilization (t/m$^3$, m, and dimensionless).
\item[$S_{\mathrm{restart}}$] Cost of one shutdown, restart, or recovery event per affected load capacity (CNY/kW/event).
\item[$C$] Capacity required for baseline demand; $C_{\mathrm{available}}$ is total available capacity.
\item[$D$] Fixed demand, written as $D=CT$.
\item[$r_C$] Existing overcapacity ratio, $(C_{\mathrm{available}}-C)/C_{\mathrm{available}}$; under the definition used here, $\rho=r_C$.
\item[$T_{\mathrm{run}}$] Runtime required to satisfy demand (h).
\item[$v_e$] Product value per unit of electricity consumption, $\lambda_{\mathrm{product}}/e_{\mathrm{intensity}}$ (CNY/kWh).
\end{description}

\paragraph{Subscripts and modifiers.}
\begin{description}[style=multiline,leftmargin=3.5cm,labelwidth=3.2cm]
\item[$\mathrm{battery}$] Battery storage; $\mathrm{h2}$ denotes power-to-hydrogen-to-power storage.
\item[$\mathrm{ely}$] Electrolyzer; $\mathrm{fc}$ and $\mathrm{store}$ denote the fuel cell and underground hydrogen store.
\item[$E$, $P$] Battery energy and power capacity bases when attached to $K$; $C$ and $\mathrm{IT}$ denote industrial-load and data-center capacity bases.
\item[$\mathrm{charge}$] Low-price charging; $\mathrm{run}$ denotes make-up production or computing operation.
\item[$\mathrm{top}$] Selected highest-price hours; $\mathrm{bottom}$, $\mathrm{high}$, and $\mathrm{low}$ denote lowest-price, high-price, and low-price hours. $H$ and $W$ give the selection length and window.
\item[$\mathrm{cap}$] Added-capacity cost; $\mathrm{wc}$, $\mathrm{wh}$, and $\mathrm{restart}$ denote working-capital, warehousing, and process-disruption costs.
\item[$\mathrm{no\ excess}$] Baseline production or computing capacity is fully used, so added capacity is required.
\item[$\mathrm{sunk}$] Idle capacity already exists and its fixed investment is treated as sunk.
\item[$\mathrm{dc}$] Data-center case; $\mathrm{Al}$ and $\mathrm{steel}$ denote the two industrial cases.
\item[$\mathrm{yr}$] Annual quantity; $\mathrm{rt}$ denotes round-trip efficiency.
\end{description}

\section{Introduction}\label{sec:introduction}

Power systems with high shares of variable renewable generation need resources that can balance supply and demand across multiple timescales~\cite{kondziella2016flexibility,denholm_challenges_2021,cole_quantifying_2021}. Energy storage and demand flexibility\footnote{We use \emph{demand flexibility} for the broad consumer-side capability of adjusting electricity use across timescales. \emph{Demand response} is reserved for event-based or market programs as used in the literature we cite. The analysis itself focuses on \emph{load shifting} with fixed total output, not curtailment without later recovery.} address this problem in different ways. Energy storage moves electricity from one period to another. A factory or data center instead moves production or computing activity, while still meeting demand for its product or service. As batteries become cheaper, competition between energy storage and demand flexibility raises a practical question. Which flexible loads can still compete with energy storage, and over what timescales?

Reviews and planning studies have documented substantial technical potential from industrial demand flexibility and have developed ways to represent demand response capabilities and willingness~\cite{golmohamadi_demand-side_2022,wei_generation_2025}. Most discussions, however, implicitly treat the underlying flexible capacity as already available. They therefore focus on the additional costs and barriers to eliciting a response, including weak incentives, limited market access, inadequate metering, automation costs, and a lack of trust~\cite{iea2025_value_demand_flexibility,iea2026_scaling_up_demand_flexibility}. Recent levelized-cost work similarly compares energy storage with direct load-control schemes such as vehicle-to-grid, smart charging, and heat-pump control by accounting for enabling equipment, consumer participation payments, and resource availability~\cite{thran_levelised_2026}. This perspective is appropriate when controllable or redundant capacity already exists, but it does not capture the capital cost of creating flexibility in a large industrial or computing load that normally operates near full capacity. Such a facility may need additional production equipment, servers, and supporting infrastructure to reduce electricity use in one period and recover output later. It may also delay sales, tie up working capital in inventory, and incur warehousing or process-disruption costs. Here, we argue that the cost of flexibility depends first on whether the required productive capacity must be newly built or is already available as sunk excess capacity.

In this article, building on the levelized-cost comparison of energy storage and flexible demand~\cite{brandt2021_blow_wind_blow,thran_levelised_2026}, we complete the cost picture for large industrial and computing loads. Brandt et al.~\cite{brandt2021_blow_wind_blow} provide the levelization idea of dividing annualized costs by electricity supplied or avoided during high-price hours. Thr{\"a}n et al.~\cite{thran_levelised_2026} extend that perspective to demand response through a levelised cost of demand response that accounts for enabling and participation costs when end-use capacity already exists. For factories and data centers that normally run near full capacity, however, the missing pieces are the incremental costs of flexible relative to baseline operation, including productive-capacity investment, working-capital financing, warehousing, and process disruption. We derive levelized cost of peak energy (LCPE) expressions for battery and hydrogen storage and for load shifting by industrial facilities and data centers, distinguishing capacity built for flexible operation from excess capacity that already exists and is treated as sunk.

Through analytical scaling relations, we show that the costs of energy storage and demand flexibility change differently with the flexibility period. Energy-storage capacity serves less annual shifted energy as the number of cycles falls, so its levelized capacity cost rises roughly in proportion to the length of the flexibility period. For load shifting under a fixed avoided-hour share, the added-capacity term is approximately independent of that period after levelization, whereas working-capital cost grows with the period when production and sale are separated for longer, and event-based restart cost is diluted as the shifted duration within each event increases.

Using representative technology costs, we then illustrate these relations numerically. That the capacity term does not rise with the flexibility period does not make new load-side capacity inexpensive. Production equipment, servers, and supporting infrastructure can still require investments that typical electricity-cost savings do not recover. Batteries therefore tend to compare favorably with frequently cycled, short-duration load shifting when new production or computing capacity would otherwise be required. Demand flexibility can remain competitive at longer durations when it can use existing excess capacity and when inventory, delay, and process-disruption costs remain limited. The comparison between batteries and demand flexibility must therefore be made for a specified flexibility period, price pattern, capacity state, and production or service process.

The remainder of this paper is organised as follows. Section~\ref{sec:observations} sets out the observations that motivate the comparison. Section~\ref{sec:technical_methods} formulates the levelized cost and the incremental cost components used here. Sections~\ref{sec:storage_lcpe} and~\ref{sec:flexible_production} then derive the energy-storage and load-shifting cost formulations, respectively. Section~\ref{sec:case_study} presents the case-study inputs and discusses the numerical comparison across flexibility timescales.

\section{Observations in recent demand response practice and research}\label{sec:observations}

\begin{table}[H]
  \centering
  \caption{Key observations from recent demand-flexibility practice and research.}
  \label{tab:practical_insights}
  \footnotesize
  \begin{tabularx}{\textwidth}{>{\raggedright\arraybackslash}p{0.04\textwidth}>{\raggedright\arraybackslash}X}
    \toprule
    1. & \textbf{Direct load shedding is generally less attractive than load shifting when lost output must be valued.} Rotational load shedding directly reduces product output, incurring an opportunity cost equal to the product's value-added, which can substantially exceed electricity costs. Load shifting avoids this loss if output can be recovered later. \\
    2. & \textbf{Load shifting is often constrained by the capital costs of redundant capacity.} To provide flexible consumption without sacrificing product output, firms must build additional capacity to compensate for reduced production during high-price periods. The resulting energy cost savings are often insufficient to offset these capital expenditures. \\
    3. & \textbf{Working capital cost is a timescale-dependent hidden cost of load shifting.} Longer-duration shifting decouples the timing of production and sales, tying up cash in unsold products. This financing cost depends on the product value per unit of electricity consumed and the firm's opportunity cost of capital. \\
    4. & \textbf{Warehousing and labor costs are not the largest cost components in the cases examined here.} Under the assumptions used below, including temporary warehouse rental and full payroll during idle periods, these incremental costs are smaller than capacity and working-capital costs. \\
    5. & \textbf{Long-duration demand flexibility may be more competitive in some applications.} Frequent adjustment can be unsuitable for continuous-process industries and faces stronger competition from low-cost batteries. Seasonal shifting, by contrast, competes with long-duration energy storage and firm clean generation. \\
    \bottomrule
  \end{tabularx}
\end{table}

First, direct load shedding is generally less attractive than load shifting when lost output must be valued. In traditional power system contexts, demand response often implicitly means that consumers help balance power demand and supply by implementing rotational curtailments; large aluminum loads, for example, have been evaluated as providers of reliability services~\cite{todd2008providing}. However, if a facility reduces its electricity consumption without recovering the lost output later, the cost of avoiding electricity use is the value-added of the undelivered product or service, which can substantially exceed electricity costs. Direct load shedding is therefore most relevant to reliability emergencies or other periods of unusually high system value. For most large loads, a more viable form of flexibility is load shifting, i.e., increasing electricity consumption before or after constrained periods while maintaining the same total output.

Second, shifting loads generally requires redundant capacity at some stage of the production or service chain. A facility operating at full capacity cannot shift production and maintain the same total output unless it can overproduce during lower-priced periods. The required capacity rises because the same output must be produced in fewer operating hours. However, deliberately creating this flexibility requires additional capital investment that may not be recovered through electricity-cost savings under typical price spreads. Our recent aluminum-smelting study shows that substantial overcapacity can be cost-effective when a shift toward recycled aluminum makes part of the smelting fleet redundant for structural reasons~\cite{lyu_industrial_overcapacity_accepted}. Because the fixed capital costs of these assets are already sunk, the economics of providing flexibility change.

Third, load shifting can incur material working capital costs, particularly over longer timescales. For intra-day or daily shifting, the required inventory adjustment is often small relative to normal operating stocks, so the incremental financing cost may be limited. Over monthly or seasonal horizons, however, shifting production to low-price periods while sales remain aligned with customer demand forces firms to hold larger inventories, delaying cash conversion and locking up capital. The induced cost depends on the product value per unit of electricity consumed and the firm's opportunity cost of capital. This barrier is most pronounced in sectors with high product values or high corporate hurdle rates.

Fourth, in the cases examined here, incremental warehousing and labor costs are smaller than capacity and working-capital costs. Physical warehousing, idle labor management, and seasonal workforce coordination still require explicit accounting. In our aluminum case study, the value of demand flexibility offsets the additional warehousing and labor costs~\cite{lyu_industrial_overcapacity_accepted}. This result is consistent with the relatively small labor share reported for modern primary aluminum smelting and other capital-intensive industries~\cite{smm_costs_2024}.

Fifth, long-duration demand flexibility may be more competitive in some applications, whereas short-duration response faces stronger competition from low-cost batteries. Short-duration demand response can also be technically risky for continuous-process industries. Studies of aluminum electrolysis show that flexible regulation must respect mass and thermal dynamics~\cite{wong_studies_2023,shen_improved_2025}, while evidence on potline shutdowns and restarts documents additional cost, operational effort, and losses in cell life~\cite{driscoll_economics_2016,lukin_reduction_2016,tabereaux_loss_2016}. Seasonal shifting presents different economic dynamics because its alternatives include long-duration energy storage and firm low-carbon generation~\cite{sepulveda_role_2018,jenkins_long-duration_2021}.

These observations motivate the total-cost comparison below.

\section{The total costs of flexibility}\label{sec:technical_methods}

This section develops a complete mathematical formulation of flexibility cost. We define shifted energy, system value, and the incremental cost of flexible relative to baseline operation, including productive-capacity investment, working capital, warehousing, and restart costs where relevant. Unless otherwise stated, all costs and values are normalized by electricity supplied, avoided, or shifted during high-price hours and reported in CNY/kWh.

\subsection{The flexible operation and period}

Flexibility is implemented by moving electricity consumption or supply capability away from high marginal-price hours and toward low marginal-price hours. Energy storage charges or produces hydrogen during low-price hours and discharges during high-price hours. Industrial loads reduce load or stop production during high-price hours and make up production during other low-price or off-peak hours. We treat one such paired action (charge and discharge for energy storage, or reduced and increased use relative to the baseline load for demand flexibility) as one flexibility period with the length of $T$. Let $N_{\mathrm{cycle}}$ denote the number of times this flexibility period repeats within a year.
\begin{equation}
  N_{\mathrm{cycle}} = \frac{8760}{T}.
\end{equation}
Let $H$ denote the number of high-price hours actually shifted, avoided, or discharged within each cycle. If $p_{\mathrm{flex}}$ is the flexible power or load capacity available for shifting, and $e_1$ is the shifted or high-price available energy in each cycle, then
\begin{equation}
  e_1 = p_{\mathrm{flex}} H.
\end{equation}
Annual shifted energy $E_{\mathrm{shift,yr}}$ equals the shifted energy per cycle multiplied by the number of annual cycles.
\begin{equation}
  E_{\mathrm{shift,yr}} = e_1 N_{\mathrm{cycle}}.
\end{equation}

System value $V_{\mathrm{system}}$ is the average marginal price during the high-price hours. Low-price charging and make-up electricity enter the cost side rather than being subtracted from that value, so that technologies with different charging or make-up electricity costs remain comparable against a common system-value reference.

The focus of this study is the relative cost of different flexibility resources. System value is reported only as a reference benchmark. For simplicity, we represent it in the case study by the average peak-hour price observed in the electricity market. That choice embeds an efficient-market interpretation of marginal prices, but the assumption affects only the cost-to-value ratio, not the ranking of technologies by resource-side cost.

Typical high-price durations $H$ also differ across flexibility periods $T$. For energy storage, that pairing changes the required power-to-energy capacity ratio. The case-specific $(T,H)$ pairs used below are given in Section~\ref{sec:case_study}.

The load-side formulation is written for large loads whose baseline electricity use is approximately steady at the timescales of interest, as in many continuous industrial processes. Household or other strongly time-varying end-use profiles are outside the main scope. Computing loads such as data centers are not literally flat. Their power draw fluctuates with workload. The mechanism we capture is nonetheless the same. Once installed servers are fully used, that occupied capacity has little room to shift work without delaying service, so flexible operation again requires spare capacity. At the hourly timescale typical of many demand-response events, data-center load variation is also relatively limited compared with day-to-day or seasonal workload swings. Figure~\ref{fig:flex_schematic} therefore starts the load panel from a flat baseline power $P_0$ and shows the distinct physical buffers used by energy storage and load shifting under the cycle definitions in Eqs.~(1)--(3).
\begin{figure}[H]
  \centering
  \includegraphics[width=0.98\textwidth,height=0.58\textheight,keepaspectratio]{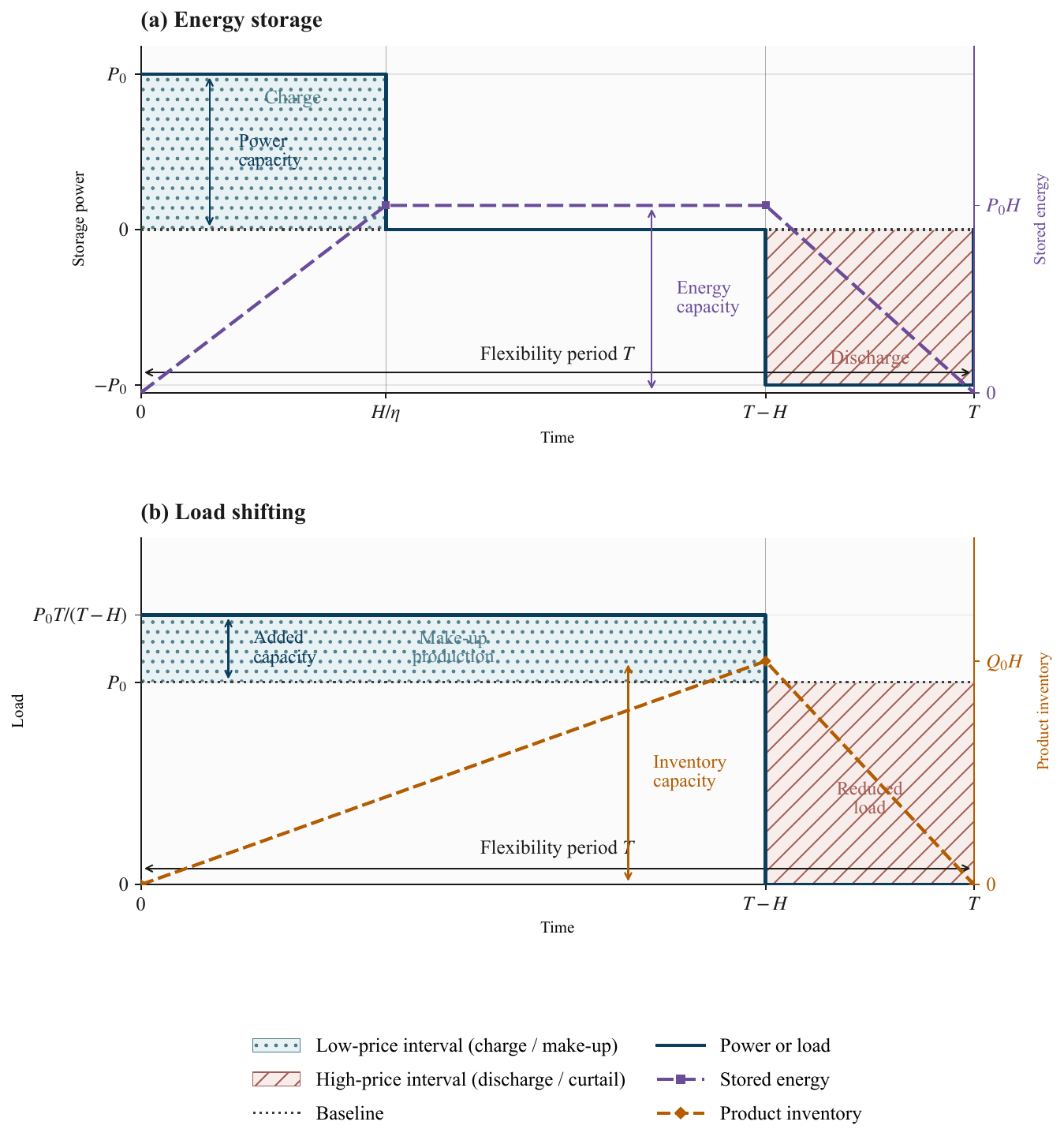}
  \caption{\textbf{(a)} Energy storage moves electricity through power and energy capacity across a flexibility period of length $T$, with high-price discharge duration $H$ and round-trip efficiency $\eta$. \textbf{(b)} Load shifting for a large load with an approximately flat baseline $P_0$ recovers the same product or service output through added productive capability and, where relevant, product inventory. Hatched bands mark low- and high-price intervals. The two panels use the same $(T,H)$ cycle definitions so that storage and load shifting can be compared on a common service basis.}
  \label{fig:flex_schematic}
\end{figure}

\subsection{The marginal costs of flexible operation}

Two principles guide the cost accounting. First, the relevant costs are those incurred by departing from baseline operation in order to provide flexibility, meaning additional investment and operating costs that would not arise under the baseline schedule. Costs that are already sunk under baseline operation are excluded, consistent with the economic treatment of sunk costs. Second, these incremental costs are expressed as a life-cycle levelized cost per unit of shifted energy, so that energy storage and load shifting can be compared on a common metric. Working capital, warehousing, and restart costs enter the numerator as distinct incremental terms where relevant.

For a given load, let $C_0$ denote the total cost under baseline operation, $C_{\mathrm{flex}}$ denote the total cost under flexible operation, and $E_{\mathrm{shift,yr}}$ denote annual shifted energy. The levelized incremental flexibility cost is
\begin{equation}
  c_{\mathrm{flex}}
  =
  \frac{
    C_{\mathrm{flex}}-C_0
  }{
    E_{\mathrm{shift,yr}}
  }.
\end{equation}
where $C_{\mathrm{flex}}-C_0$ is the incremental annualized cost of flexible operation relative to baseline operation. It comprises low-price charging or make-up electricity $\Delta C_{\mathrm{op}}$, annualized power-, energy-, production-, or computing-capacity investment $\Delta C_{\mathrm{cap}}$, working-capital carrying cost of physical-product inventory $\Delta C_{\mathrm{wc}}$, and restart or other process-disruption cost $\Delta C_{\mathrm{aux}}$.
\begin{equation}
  \label{eq:incremental_components}
  C_{\mathrm{flex}}-C_0
  =
  \Delta C_{\mathrm{op}}
  +
  \Delta C_{\mathrm{cap}}
  +
  \Delta C_{\mathrm{wc}}
  +
  \Delta C_{\mathrm{aux}}.
\end{equation}
As noted above, the avoided high-price electricity value is kept outside $C_{\mathrm{flex}}-C_0$ and reported separately as the system-value reference $V_{\mathrm{system}}$. Enabling items such as communication, metering, controls, and retrofit are not a separate omitted category. Capital outlays required for flexible operation enter $\Delta C_{\mathrm{cap}}$; other items are assigned to the operating or disruption term to which they belong.

\subsection{Capital-cost annualization}

Following the life-cycle costing approach used by Brandt et al.~\cite{brandt2021_blow_wind_blow} and related levelized-cost studies, capital outlays are first converted to equivalent annual costs and then divided by annual shifted energy to obtain a levelized unit cost. For an original investment $I$, asset lifetime $y$, and discount rate $r_{\mathrm{int}}$, the capital recovery factor is
\begin{equation}
  \mathrm{CRF} =
  \frac{
    r_{\mathrm{int}}(1+r_{\mathrm{int}})^y
  }{
    (1+r_{\mathrm{int}})^y - 1
  }.
\end{equation}
When $r_{\mathrm{int}}=0$, $\mathrm{CRF}=1/y$.

If fixed O\&M is a fraction $f_{\mathrm{FOM}}$ of investment, the equivalent annualized capacity cost is
\begin{equation}
  K = I\left(\mathrm{CRF}+f_{\mathrm{FOM}}\right).
\end{equation}
If fixed O\&M is instead an absolute annual amount $\mathrm{FOM}_{\mathrm{abs}}$, then
\begin{equation}
  K = I\cdot\mathrm{CRF}+\mathrm{FOM}_{\mathrm{abs}}.
\end{equation}
Subsequent $K$ terms already include capital recovery and fixed O\&M. Technology-specific investment costs, lifetimes, discount rates, and O\&M assumptions are reported in the case study.

\section{LCPE of Energy-storage}\label{sec:storage_lcpe}

This section applies the levelized incremental flexibility cost to energy storage. Capacity investment is paid once and recovered over the annual cycles $N_{\mathrm{cycle}}=8760/T$. As the flexibility period lengthens, the same power and energy capacity serves less annual shifted energy $E_{\mathrm{shift,yr}}=e_1N_{\mathrm{cycle}}$, so the levelized capacity cost rises roughly in proportion to the period length. Low-price charging electricity remains on the cost side, as noted above. The expressions below are written for batteries and hydrogen storage; other storage technologies in the case study use the same structure with their own annualized capacity costs and efficiencies.

\subsection{Battery storage}

Battery cost comprises energy-capacity, power-capacity, and charging-electricity terms. Let $K_E$ and $K_P$ be the equivalent annualized energy- and power-capacity costs. The flexibility period $T$ sets the number of cycles per year, while the high-price discharge duration $H$ sets the power capacity needed for each event.

For a cycle that supplies $e_1$ during high-price hours, the required battery energy capacity, charging electricity, and power capacity are
\begin{equation}
  E_{\mathrm{bat}} = e_1,
\end{equation}
\begin{equation}
  E_{\mathrm{charge}} =
  \frac{e_1}{\eta_{\mathrm{rt}}},
\end{equation}
\begin{equation}
  p_{\mathrm{bat}} =
  \frac{e_1}{H}.
\end{equation}
For comparability across technologies, charging and discharging power capacities are set equal, with discharge power sized by $H$ and efficiency losses reflected in charging energy (and therefore charging time) at that power. The calculation does not optimize a technology-specific charging duration or an asymmetric converter configuration. Separate sizing of the two directions can lower investment for some storage technologies, so the reported costs may overstate individually optimized designs. The common convention is kept so that the comparison does not embed a different operating strategy in each technology.

With annual shifted energy $E_{\mathrm{shift,yr}}=e_1N_{\mathrm{cycle}}$ and low-price charging price $\lambda_{\mathrm{low}}$, the annual battery cost is
\begin{equation}
  C_{\mathrm{battery,yr}} =
  e_1 K_E
  +
  \frac{e_1 K_P}{H}
  +
  N_{\mathrm{cycle}}
  \frac{e_1 \lambda_{\mathrm{low}}}{\eta_{\mathrm{rt}}}.
\end{equation}
Dividing by annual shifted energy and substituting $N_{\mathrm{cycle}}=8760/T$ gives
\begin{equation}
  c_{\mathrm{battery}} =
  \frac{C_{\mathrm{battery,yr}}}{e_1N_{\mathrm{cycle}}}
  =
  \frac{\lambda_{\mathrm{low}}}{\eta_{\mathrm{rt}}}
  +
  \underbrace{
    \frac{T}{8760}
    \left(
      K_E
      +
      \frac{K_P}{H}
    \right)
  }_{\text{capacity cost rising with \(T\)}}.
\end{equation}
The first term is charging electricity. The second is annualized capacity cost and grows with the flexibility period. If the battery cycles only once per year, the same battery and inverter costs are allocated to a single shifted-energy event; if it cycles daily, those costs are spread over hundreds of events.

\subsection{Hydrogen storage}

Hydrogen storage is represented as a power-to-hydrogen-to-power chain with electrolysis, hydrogen storage, and fuel-cell reconversion. Let $K_{\mathrm{ely}}$, $K_{\mathrm{fc}}$, and $K_{\mathrm{store}}$ be the corresponding equivalent annualized capacity costs. Lower conversion efficiency raises the electricity-purchase term, while inexpensive energy capacity can soften the rise of capital cost at long duration relative to batteries. Annual cycling still governs capacity utilization in the same way as for batteries.

As above, the central convention sets electrolyzer and fuel-cell power equal and sized by $H$, rather than optimizing an asymmetric charging window. Conversion losses enter purchased electricity and hydrogen energy. The required capacities are
\begin{align}
  k_{\mathrm{ely}} &=
  \frac{e_1}{H}, \\
  k_{\mathrm{fc}} &=
  \frac{e_1}{H}, \\
  e_{\mathrm{h2}} &=
  \frac{e_1}{\eta_{\mathrm{fc}}}.
\end{align}
The annual hydrogen-storage cost is
\begin{equation}
  C_{\mathrm{h2,yr}} =
  k_{\mathrm{ely}} K_{\mathrm{ely}}
  +
  k_{\mathrm{fc}} K_{\mathrm{fc}}
  +
  e_{\mathrm{h2}} K_{\mathrm{store}}
  +
  N_{\mathrm{cycle}}
  \frac{e_1 \lambda_{\mathrm{low}}}{\eta_{\mathrm{ely}}\eta_{\mathrm{fc}}}.
\end{equation}
Dividing by annual shifted energy and substituting $N_{\mathrm{cycle}}=8760/T$ yields
\begin{equation}
  c_{\mathrm{h2}} =
  \frac{C_{\mathrm{h2,yr}}}{e_1N_{\mathrm{cycle}}}
  =
  \frac{\lambda_{\mathrm{low}}}{\eta_{\mathrm{ely}}\eta_{\mathrm{fc}}}
  +
  \underbrace{
    \frac{T}{8760}
    \left(
      \frac{K_{\mathrm{ely}}}{H}
      +
      \frac{K_{\mathrm{fc}}}{H}
      +
      \frac{K_{\mathrm{store}}}{\eta_{\mathrm{fc}}}
    \right)
  }_{\text{capacity cost rising with \(T\)}}.
\end{equation}
As for batteries, annualized capacity costs are diluted by annual shifted energy and therefore rise with the flexibility period at fixed $H$.

The factor $T/8760$ is a utilization effect distinct from round-trip efficiency losses. Improving efficiency reduces the electricity-purchase term but does not remove this factor. Long-duration competitiveness therefore requires low capacity costs relative to the service provided, not only higher conversion efficiency.

\section{Load-shifting LCPE}\label{sec:flexible_production}

This section applies the levelized incremental flexibility cost to load shifting with fixed total output. As discussed in Section~\ref{sec:observations}, direct shedding without later recovery is closer to rotational curtailment. Lost output is then valued at the product or service value-added and can far exceed electricity costs, so it is unlikely to be the mainstream form of large-load flexibility except in reliability emergencies.

As noted in Section~\ref{sec:technical_methods}, the derivation focuses on large loads with approximately steady baseline electricity use at the timescales considered here. Total product or service demand over each flexibility period must still be met, so output reduced in high-price hours is restored by make-up production or computing in other hours. Physical-product inventory reconciles production with sales where relevant.

Relative to energy storage, the distinctive terms are productive-capacity investment, working-capital financing, warehousing, and process disruption. The subsections below derive these costs under no-excess and sunk-excess capacity. Scaling relations are collected afterward.

\subsection{No-excess-capacity case}
\label{sec:no_excess}

In the no-excess setting, baseline production or computing capacity is already fully used. Avoiding high-price hours while keeping total output therefore requires higher operating intensity in other hours and, when no idle capability exists, additional production lines, furnaces, servers, cooling, or power infrastructure.

\subsubsection{Industrial productive-capacity investment}

For industrial loads that make up physical product, productive-capacity investment covers production equipment, fixed labor, and other fixed costs configured with available capacity. The capacity boundary must match the physical configuration. If flexibility requires linked stages to expand together, $K_{\mathrm{tonne/year}}$ or $K_C$ represents the integrated block; if intermediate stages are temporally decoupled, the same formula can be applied only to the component that actually expands. Existing oversized or decoupled stages are treated as the sunk-excess case below rather than as no-excess.

To place industrial capacity in electricity-shifting units, convert the equivalent annualized cost per unit of annual production capacity into an equivalent annualized cost per kilowatt of electrical load capability, $K_C$. With electricity intensity $e_{\mathrm{intensity}}$,
\begin{equation}
  K_C =
  \frac{K_{\mathrm{tonne/year}}}
  {e_{\mathrm{intensity}}/8760}.
\end{equation}
Let $p_0$ be baseline load capability, $D=p_0T$ the demand in one flexibility period, $e_1$ the shifted high-price energy, and $\Delta t=e_1/p_0$ the avoided high-price duration. Make-up runtime and required added power are
\begin{equation}
  T_{\mathrm{run,no\ excess}} =
  T-\Delta t,
  \qquad
  \Delta p =
  \frac{e_1}{T_{\mathrm{run,no\ excess}}}.
\end{equation}
The annualized capital cost of this added capacity, divided by annual shifted energy $E_{\mathrm{shift,yr}}=e_1N_{\mathrm{cycle}}$, is
\begin{equation}
  c_{\mathrm{cap,no\ excess}} =
  \frac{K_C\Delta p}{E_{\mathrm{shift,yr}}}
  =
  \frac{K_C}
  {N_{\mathrm{cycle}}T_{\mathrm{run,no\ excess}}}.
\end{equation}
Added power is set by the make-up duration within one cycle; more annual cycles dilute the same added capacity over more shifted energy. When $e_1/p_0$ is small, $T_{\mathrm{run,no\ excess}}\approx T$ and
\begin{equation}
  c_{\mathrm{cap,no\ excess}}
  \approx
  \frac{K_C}{8760},
\end{equation}
which does not depend on the flexibility period. Load-side added capacity therefore does not share the energy-storage rise with period length. It mainly raises operating intensity in remaining hours, whereas storage capacity must bridge a longer intertemporal transfer.

\subsubsection{Data-center capacity investment}

For data centers, the corresponding investment is additional IT load with supporting power, cooling, and service capability. Without idle capability, reducing electricity use in high-price hours while maintaining service output requires extra computing capacity for later make-up. Data centers need not hold physical inventory, but delay or service-quality losses can create a separate time cost and should be added when parameterized. With annualized computing-capacity cost $K_{\mathrm{IT}}$,
\begin{equation}
  c_{\mathrm{cap,dc}} =
  \frac{K_{\mathrm{IT}}}
  {N_{\mathrm{cycle}}T_{\mathrm{run,no\ excess}}}.
\end{equation}

\subsubsection{Working-capital financing cost}

Working-capital carrying cost is the opportunity cost of product value locked in inventory when output is produced early. It is distinct from physical warehousing. Longer flexibility periods typically lengthen average lock-up, so the term matters especially for long-duration aluminum and steel shifting. With inventory value $\lambda_{\mathrm{product}}$, financing rate $r_{\mathrm{discount}}$, lock-up duration $\tau_{\mathrm{wc}}$, and electricity intensity $e_{\mathrm{intensity}}$,
\begin{equation}
  c_{\mathrm{wc}} =
  \frac{
    \lambda_{\mathrm{product}} \cdot r_{\mathrm{discount}} \cdot \tau_{\mathrm{wc}}
  }{e_{\mathrm{intensity}}}.
\end{equation}
When the relative production--sales shape is held fixed, $\tau_{\mathrm{wc}}$ rises with the flexibility period. Section~\ref{sec:load_scaling} formalizes this linear growth.

\subsubsection{Physical warehousing cost}

Physical warehousing plays a role analogous to energy capacity in storage. It covers warehouse space, yards, material handling, and inventory management, not the financing of inventory value. With production and demand time series, inventory $\mathrm{Inv}(t)$, stacking density $\rho_{\mathrm{material}}$, effective height $h_{\mathrm{effective}}$, and space utilization $\eta_{\mathrm{storage}}$,
\begin{align}
  \mathrm{Inv}(t) &=
  \mathrm{cumsum}(\mathrm{Production})
  -
  \mathrm{cumsum}(\mathrm{Demand}), \\
  \mathrm{Area}(t) &=
  \frac{\mathrm{Inv}(t)}
  {\rho_{\mathrm{material}} \cdot h_{\mathrm{effective}} \cdot \eta_{\mathrm{storage}}}, \\
  \mathrm{Warehousing\ cost} &=
  \mathrm{Area} \cdot \mathrm{rental\ rate} \cdot \mathrm{time}.
\end{align}
When a tonne-month proxy is used instead, with warehousing rate $w_{\mathrm{wh}}$ and inventory duration $m(T)$,
\begin{equation}
  c_{\mathrm{wh}} =
  \frac{w_{\mathrm{wh}}m(T)}{e_{\mathrm{intensity}}}.
\end{equation}
Sector-specific $w_{\mathrm{wh}}$, $m(T)$, and $e_{\mathrm{intensity}}$ are case-study inputs.

\subsubsection{Process-disruption cost}

Restart and recovery costs arise from deviations from stable operation, including shutdowns, restarts, thermal or chemical recovery, short interruptions, equipment-life loss, efficiency loss, and quality loss. They are incremental to baseline operation and are closer to power-capacity costs than to energy costs, because one event affects a block of equipment and is then allocated to avoided electricity. Controlled partial modulation within a stable operating range may avoid a full shutdown, whereas complete interruption can break the process state on which continuous production depends.

If one event costs $S_{\mathrm{restart}}$ per affected load capacity (CNY/kW/event) and the high-price reduction duration is $H$, the affected power is $p_{\mathrm{restart}}=e_1/H$ and
\begin{equation}
  c_{\mathrm{restart}} =
  \frac{S_{\mathrm{restart}}}{H}.
\end{equation}
Annual event cost and annual shifted energy both scale with $N_{\mathrm{cycle}}$, so the cycle count cancels. Multiple equivalent events per cycle can multiply $S_{\mathrm{restart}}$. Because the term falls as $1/H$, a fixed event cost becomes smaller per shifted kilowatt-hour as the shifted duration lengthens. Event costs and process boundaries are case-study inputs.

\subsubsection{Aggregate no-excess cost}

Summing the incremental terms and dividing by annual shifted energy gives the no-excess LCPE for physical-product loads,
\begin{equation}
  c_{\mathrm{no\ excess}} =
  \lambda_{\mathrm{run,no\ excess}}
  +
  c_{\mathrm{cap,no\ excess}}
  +
  c_{\mathrm{wc}}
  +
  c_{\mathrm{wh}}
  +
  c_{\mathrm{restart}},
\end{equation}
where $\lambda_{\mathrm{run,no\ excess}}$ is the average make-up electricity price over $T_{\mathrm{run,no\ excess}}$. For a data center without parameterized inventory or restart,
\begin{equation}
  c_{\mathrm{no\ excess,dc}} =
  \lambda_{\mathrm{run,no\ excess}}
  +
  c_{\mathrm{cap,dc}}.
\end{equation}
Delay, service-quality, or contractual costs should be added when they are part of the selected case boundary.

\subsection{Sunk-excess-capacity case}
\label{sec:sunk_excess}

In the sunk-excess setting, idle production or computing capacity already exists, so the added-capacity term is set to zero. Make-up electricity, warehousing, working capital, and restart or recovery costs remain where applicable.

For physical-product loads,
\begin{equation}
  c_{\mathrm{sunk}} =
  \lambda_{\mathrm{run,sunk}}
  +
  c_{\mathrm{wc}}
  +
  c_{\mathrm{wh}}
  +
  c_{\mathrm{restart}},
\end{equation}
with $\lambda_{\mathrm{run,sunk}}$ the average electricity price in hours that use idle capacity. For a data center with sunk excess capacity and no separately parameterized delay cost,
\begin{equation}
  c_{\mathrm{sunk,dc}} =
  \lambda_{\mathrm{run,sunk}}.
\end{equation}

\subsubsection{Linking overcapacity to avoidable high-price hours}

Under the definition of the overcapacity ratio used here, the avoidable high-price share equals that ratio, $\rho=r_C$ (Proposition~\ref{prop:overcapacity_movable_energy} below). When that ratio is linked to price selection, the top $\rho$ share of hours in each window is treated as avoided, and the remaining $1-\rho$ share is the operating period. Then
\begin{equation}
  \lambda_{\mathrm{run,sunk}}
  =
  \mathrm{mean}
  \left(
    \lambda_t \mid
    t \notin \mathrm{top}\ \rho\ \mathrm{price\ hours}
  \right).
\end{equation}
In the no-excess case, avoided duration is $\Delta t=e_1/p_0$ and make-up runtime is $T_{\mathrm{run,no\ excess}}=T-\Delta t$, with average price $\lambda_{\mathrm{run,no\ excess}}$.

\subsection{Scaling laws for load-shifting costs}
\label{sec:load_scaling}

The cost expressions above already indicate how the main terms change with the flexibility period and shifted duration. The propositions below formalize those relations.

\begin{proposition}[The $O(1)$ scaling of added-capacity cost with the shifting period]
  \label{prop:capacity_o1}
  When the same share $\rho$ of high-price hours is avoided in each flexibility period, added-capacity cost levelized by annual shifted energy does not change with the flexibility period $T$, i.e., $c_{\mathrm{cap}}=O(1)$.
\end{proposition}

\begin{proof}
  Consider a load that must satisfy fixed demand $D=CT$ over a flexibility period of length $T$, where $D$ is total demand in that period and $C$ is the equivalent load capacity just sufficient to meet baseline demand. If flexible operation avoids a share $\rho$ of high-price hours without reducing total output, production or service must be completed within $(1-\rho)T$ hours after avoiding $\rho T$ hours. Let $C_{\mathrm{flex}}$ denote the available capacity required for flexible operation. Then
  \begin{equation}
    C_{\mathrm{flex}} =
    \frac{C}{1-\rho}.
  \end{equation}
  Let $\Delta C$ denote the required added capacity.
  \begin{equation}
    \Delta C =
    C_{\mathrm{flex}}-C
    =
    \frac{\rho}{1-\rho}C.
  \end{equation}
  Thus, at fixed avoided share $\rho$, physical capacity expansion $\Delta C$ does not increase with the flexibility period $T$. The high-price electricity shifted within one window is $E_{\mathrm{shift}}=\rho CT$. There are $N_{\mathrm{cycle}}=8760/T$ windows per year, so annual shifted energy is
  \begin{equation}
    N_{\mathrm{cycle}}E_{\mathrm{shift}}
    =
    \frac{8760}{T}\rho CT
    =
    8760\rho C.
  \end{equation}
  If the equivalent annualized cost of added capacity is $K_C$, the added-capacity cost per unit of shifted energy is
  \begin{equation}
    c_{\mathrm{cap}}
    =
    \frac{K_C \Delta C}
    {N_{\mathrm{cycle}} \rho CT}
    =
    \frac{K_C}{8760(1-\rho)},
  \end{equation}
  which is also independent of $T$.
\end{proof}

\begin{proposition}[The $O(T)$ scaling of working-capital carrying cost with shifting duration]
  \label{prop:working_capital_ot}
  If the relative timing of production and sales is stretched in proportion to the flexibility period $T$, working-capital carrying cost per unit of shifted energy increases linearly with $T$, i.e., $c_{\mathrm{wc}}=O(T)$.
\end{proposition}

\begin{proof}
  Let $v_e=\lambda_{\mathrm{product}}/e_{\mathrm{intensity}}$ denote product value per unit of electricity consumption, with $\lambda_{\mathrm{product}}$ the product price and $e_{\mathrm{intensity}}$ the electricity intensity. Let $r_{\mathrm{discount}}$ be the annual cost of capital and $\tau_{\mathrm{wc}}$ the average cash lock-up duration per unit of shifted energy in years. Then
  \begin{equation}
    c_{\mathrm{wc}} =
    v_e r_{\mathrm{discount}}\tau_{\mathrm{wc}}.
  \end{equation}
  If production--sales trajectories keep the same relative shape and are only stretched with period length, then $\tau_{\mathrm{wc}}=\gamma T/8760$ for a shape factor $\gamma$, and
  \begin{equation}
    c_{\mathrm{wc}}
    =
    \frac{\gamma v_e r_{\mathrm{discount}}}{8760}T
    =
    O(T).
  \end{equation}
\end{proof}

\begin{proposition}[Relationship between excess capacity and movable energy]
  \label{prop:overcapacity_movable_energy}
  Under the definition of overcapacity ratio used in this study, the maximum share of high-price hours that can be avoided within one period equals the existing overcapacity ratio, $\rho=r_C$, and the corresponding movable energy within one window is $E_{\mathrm{movable}}=r_CD$.
\end{proposition}

\begin{proof}
  Let $C$ denote the capacity required to meet baseline demand and $C_{\mathrm{available}}$ the total available capacity. Existing excess capacity is $C_{\mathrm{available}}-C$, and the overcapacity ratio is
  \begin{equation}
    r_C =
    \frac{C_{\mathrm{available}}-C}{C_{\mathrm{available}}}.
  \end{equation}
  Hence
  \begin{equation}
    C_{\mathrm{available}} =
    \frac{C}{1-r_C}.
  \end{equation}
  Holding total demand $D=CT$ fixed and adding no new capacity, if the facility runs using all available capacity, the runtime required to meet demand is
  \begin{equation}
    T_{\mathrm{run}}
    =
    \frac{D}{C_{\mathrm{available}}}
    =
    (1-r_C)T.
  \end{equation}
  The avoidable time is therefore $T-T_{\mathrm{run}}=r_CT$, so $\rho=r_C$. The corresponding movable energy within one window is
  \begin{equation}
    E_{\mathrm{movable}}
    =
    C(T-T_{\mathrm{run}})
    =
    r_C D.
  \end{equation}
\end{proof}

\section{Case study and discussion}\label{sec:case_study}
\label{sec:results_discussion}

This section applies the preceding formulas to battery, pumped-hydro, vanadium-flow and hydrogen storage, aluminum smelting, steelmaking, and data centers, then discusses how the comparison changes across flexibility timescales. The main text retains only the inputs needed to read the figures.

\subsection{Flexibility timescales and price data}

Four representative flexibility cases are evaluated. The labels describe the order of magnitude of the high-price hours shifted, whereas $T$ denotes the complete cycle window. Thus, hourly flexibility is modeled within a 24-h cycle and seasonal flexibility within an annual cycle.

\begin{table}[H]
  \centering
  \caption{Case-study flexibility periods and high-price reduction durations}
  \label{tab:cycle_periods}
  \footnotesize
  \begin{tabularx}{\textwidth}{>{\raggedright\arraybackslash}Xrrr}
    \toprule
    \textbf{Flexibility label} & \textbf{Period $T$} & \textbf{Shifted duration $H$} & \textbf{Cycles per year} \\
    \midrule
    hourly   & 24 h   & 2.6 h  & 365 \\
    daily    & 168 h  & 24 h   & 52 \\
    weekly   & 730 h  & 168 h  & 12 \\
    seasonal & 8760 h & 2160 h & 1 \\
    \bottomrule
  \end{tabularx}
\end{table}

Prices are observed 15-minute Guangdong day-ahead provincial-average node prices for the latest complete 365-day sample. For each complete price window $W$, the $H$ hours in Table~\ref{tab:cycle_periods} define the high-price selection length. Intervals are sorted by price and the highest $H$ hours, lowest $H$ hours, and remaining intervals are averaged to give $\overline{\lambda}_{\mathrm{top},H,W}$, $\overline{\lambda}_{\mathrm{bottom},H,W}$, and $\overline{\lambda}_{\mathrm{non\ top},H,W}$. The case-study price metrics are
\begin{align}
  V_{\mathrm{system}} &=
  \mathrm{mean}_{W}\left(\overline{\lambda}_{\mathrm{top},H,W}\right), \\
  \lambda_{\mathrm{charge}} &=
  \mathrm{mean}_{W}\left(\overline{\lambda}_{\mathrm{bottom},H,W}\right), \\
  \lambda_{\mathrm{run}} &=
  \mathrm{mean}_{W}\left(\overline{\lambda}_{\mathrm{non\ top},H,W}\right).
\end{align}
The first is the modeled system value during high-price hours. The second is used for battery charging and hydrogen production, and the third for industrial and data-center make-up operation. Make-up or charging electricity is counted on the cost side and is not deducted from system value. For the hourly, daily, weekly, and seasonal cases, the system values are 0.472, 0.489, 0.475, and 0.510~CNY/kWh.\footnote{Approximate mid-market rates as of early August~2026: 1~CNY~$\approx$~0.148~USD~$\approx$~0.130~EUR.}

\subsection{Technology and load cases}

Table~\ref{tab:case_resources} summarizes the resources compared below. Detailed investment, lifetime, efficiency, inventory, and source-mapping inputs are given in the Supplementary Data.

\begin{table}[H]
  \centering
  \caption{Case-study resources and central cost boundaries}
  \label{tab:case_resources}
  \footnotesize
  \begin{tabularx}{\textwidth}{>{\raggedright\arraybackslash}p{0.20\textwidth}>{\raggedright\arraybackslash}p{0.38\textwidth}>{\raggedright\arraybackslash}X}
    \toprule
    \textbf{Resource} & \textbf{Central boundary} & \textbf{Key notes} \\
    \midrule
    Battery & China-localized energy and inverter costs & Lifetime and efficiency from DEA~\cite{PyPSA,pypsa_costs,dea_energy_storage} \\
    Pumped hydro & China 4-h project investment & Long life; Zhen'an cross-check~\cite{nea_zhenan_2024} \\
    Vanadium flow & Chinese 100-MW/400-MWh project & Power/energy shares from PNNL~\cite{sasac_vrfb_2025,pnnl_esgc_2024} \\
    Hydrogen storage & Electrolysis, cavern storage, fuel cell & DEA technology data~\cite{dea_renewable_fuels,dea_tech_el_dh,dea_energy_storage} \\
    \midrule
    Aluminum & 13,300~kWh/t-Al; 18,000~CNY/(t-Al/yr); restart 110~CNY/kW/event & Smelter-only; full potline interruption~\cite{CHALCO2020AnnualReport,Hu2009,noauthor_high_2020,noauthor_alcoa_2021} \\
    Steel & 440~kWh/t; 20,280~CNY/(t/yr) whole plant & EAF-only retained as sensitivity~\cite{ft_stegra_boden,steelonthenet_eaf,bataille_review_2018} \\
    Data center & 77,040~CNY/kW shell/core (15~yr); 180,000~CNY/kW fit-out (5~yr) & Generation excluded~\cite{jll_data_center_outlook,irs_publication_946_2024} \\
    \bottomrule
  \end{tabularx}
\end{table}

All storage technologies use the same $T$, $H$, Guangdong charging prices, and discount rate. Load cases distinguish newly built capacity from existing capacity treated as sunk. Each resource is evaluated at its modeled flexible power and normalized by the energy shifted at that power, rather than by a technology-specific modulation depth. The tabulated values are case-study anchors, not sector-wide averages.

\subsection{Cross-timescale comparison}

\begin{figure}[!htbp]
  \centering
  \includegraphics[width=0.98\textwidth,height=0.60\textheight,keepaspectratio]{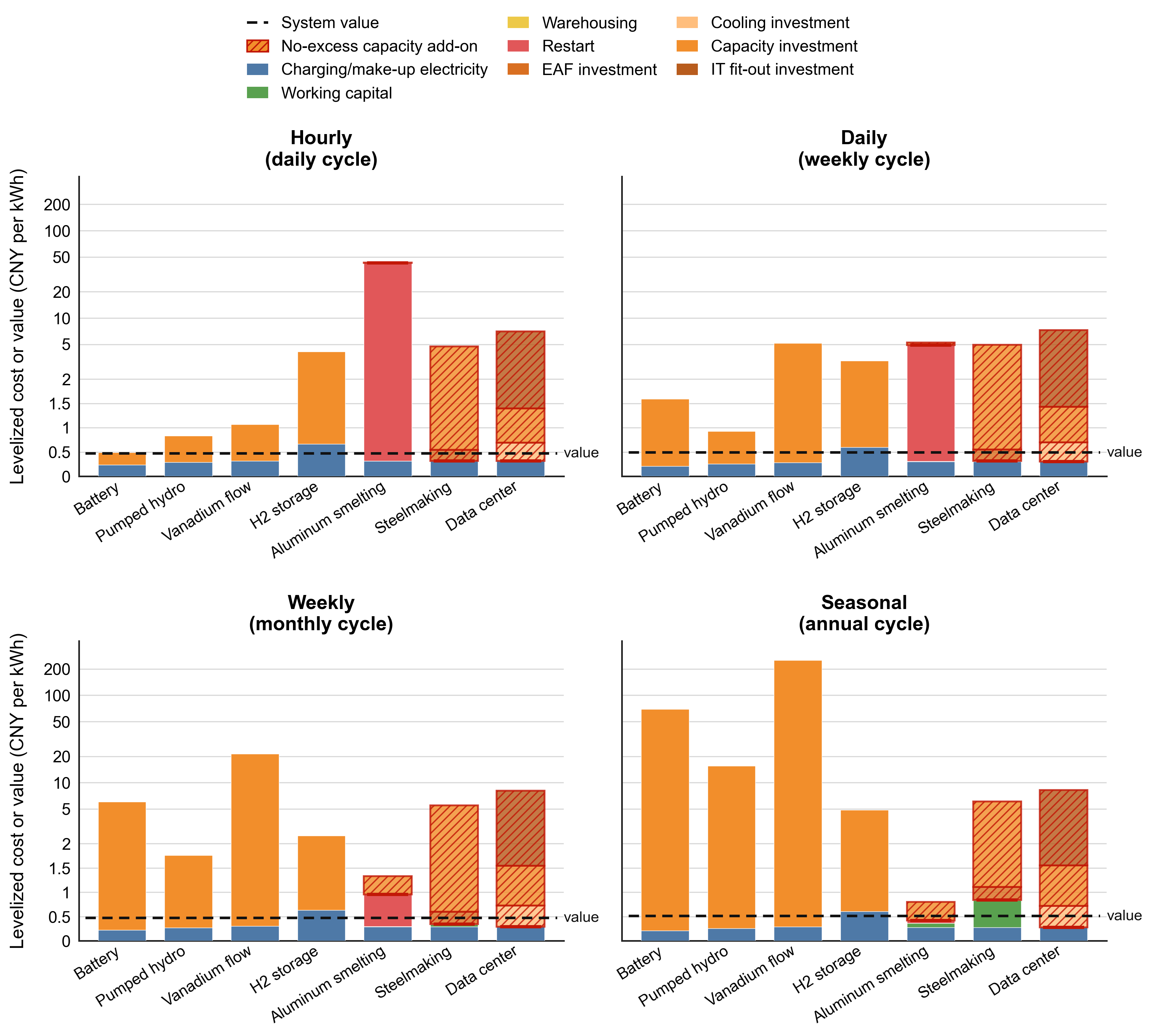}
  \caption{Levelized flexibility cost (LCPE) across four timescales for energy storage and large-load shifting. The dashed line is the Guangdong day-ahead high-price reference. Solid segments are operating, inventory, warehousing, restart, or storage capacity costs; red-hatched segments add new productive or computing capacity when no excess capacity exists. Battery and other storage costs rise with the flexibility period as annual cycling falls, whereas full-interruption aluminum falls sharply once restart cost is diluted, producing a nearly two-order-of-magnitude reversal relative to batteries between the hourly and seasonal cases. Whole-plant steel and data-center expansion remain expensive at all timescales because new capacity dominates. Values below the dashed line have modeled costs below the selected system-value reference before omitted implementation costs.}
  \label{fig:flex_costs}
\end{figure}

As shown in Fig.~\ref{fig:flex_costs} and Table~\ref{tab:compact_results}, energy-storage LCPE rises as the flexibility period lengthens and annual cycling falls, whereas load-side costs follow process-specific paths. Battery LCPE increases from 0.494~CNY/kWh in the hourly case to 69.5~CNY/kWh in the seasonal case. Full-interruption aluminum moves in the opposite direction, from 43.0 to 0.804~CNY/kWh when new smelting capacity is included, so aluminum is about 87 times more expensive than batteries at the hourly scale and about 86 times cheaper at the seasonal scale. Whole-plant steel and data-center expansion remain expensive throughout (4.75--6.09 and 7.02--8.21~CNY/kWh), though the capital cost of new productive or computing capacity does not increase with longer periods in the way battery cost does.

\begin{table}[!htbp]
  \centering
  \scriptsize
  \setlength{\tabcolsep}{4pt}
  \renewcommand{\arraystretch}{1.08}
  \caption{Compact summary of central hourly and seasonal results. Each result is reported as LCPE in CNY/kWh, with LCPE divided by the corresponding system value in parentheses.}
  \label{tab:compact_results}
  \begin{tabularx}{\textwidth}{>{\raggedright\arraybackslash}p{0.22\textwidth}>{\centering\arraybackslash}p{0.14\textwidth}>{\centering\arraybackslash}p{0.14\textwidth}>{\raggedright\arraybackslash}X}
    \toprule
    \textbf{Resource and boundary} & \textbf{Hourly LCPE (ratio)} & \textbf{Seasonal LCPE (ratio)} & \textbf{Dominant modeled driver or qualification} \\
    \midrule
    Battery & 0.494 (1.05) & 69.5 (136) & Capacity investment; annual utilization collapses at seasonal scale. \\
    Pumped hydro & 0.833 (1.76) & 15.5 (30.4) & Long life and lower energy-capacity cost moderate the investment term. \\
    Vanadium flow & 1.07 (2.26) & 252 (494) & Energy-capacity investment dominates under the Chinese project anchor. \\
    Hydrogen storage & 4.12 (8.71) & 4.85 (9.51) & Conversion and storage investment plus efficiency losses. \\
    \midrule
    Aluminum: new capacity & 43.0 (90.9) & 0.804 (1.58) & Restart cost dominates hourly response. \\
    Aluminum: sunk capacity & 42.6 (90.3) & 0.416 (0.82) & Removing investment matters primarily at seasonal scale. \\
    Steel: whole plant & 4.75 (10.1) & 6.09 (11.9) & Whole-plant expansion is the broad upper boundary. \\
    Steel: EAF only & 0.542 (1.15) & 1.11 (2.17) & Component-only expansion sharply reduces capacity cost. \\
    Steel: sunk capacity & 0.317 (0.67) & 0.841 (1.65) & Seasonal working capital remains material. \\
    Data center: new capacity & 7.02 (14.9) & 8.21 (16.1) & New IT fit-out dominates. \\
    Data center: sunk capacity & 0.316 (0.67) & 0.275 (0.54) & Lower bound excluding delay and SLA costs. \\
    \bottomrule
  \end{tabularx}
\end{table}

\begin{figure}[!htbp]
  \centering
  \includegraphics[width=0.98\textwidth,height=0.56\textheight,keepaspectratio]{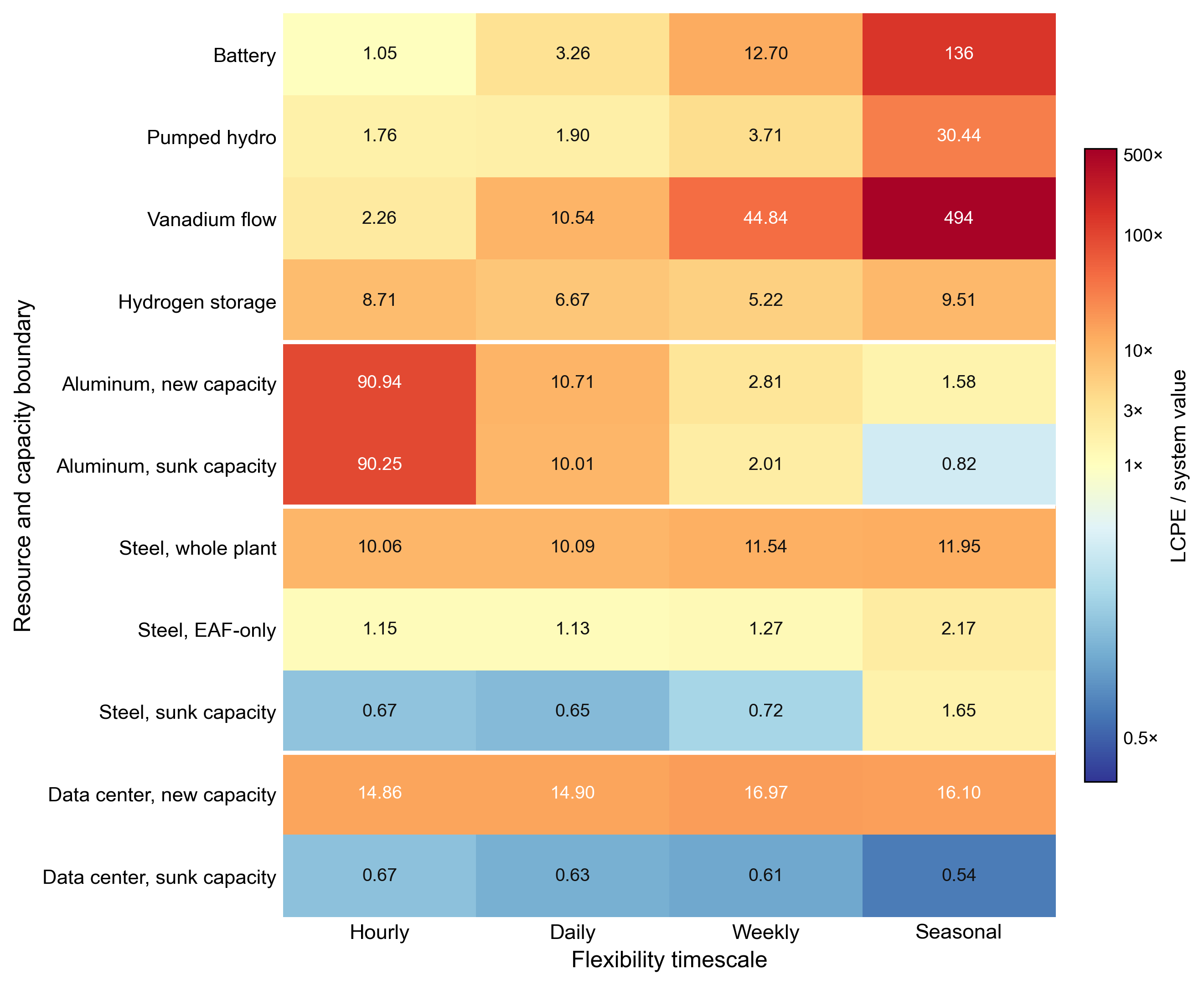}
  \caption{LCPE divided by the corresponding Guangdong Top-$H$ system value. Ratios below one mean that modeled resource cost is below the high-price reference before omitted enabling, transaction, and service-quality costs. No energy-storage technology falls below one in the central case; batteries come closest in the hourly case. Load-side ratios fall below one mainly when existing capacity is treated as sunk, and for seasonal aluminum once restart cost is diluted. White separators distinguish storage, aluminum, steel, and data-center groups.}
  \label{fig:cost_value_ratio}
\end{figure}

Fig.~\ref{fig:cost_value_ratio} compares these costs with the Guangdong high-price reference. No energy-storage technology falls below a ratio of one in the central case; batteries come closest in the hourly case (1.05). Load-side ratios fall below one mainly when existing capacity is treated as sunk, and for seasonal aluminum once restart cost is diluted. These ratios use gross system value and should not be read as private profitability. Because the plotted LCPE already omits enabling and transaction costs, a ratio above one means the resource-side lower bound alone exceeds the reference value, so there is no residual margin with which to cover those omitted items.

\subsection{New capacity versus sunk excess capacity}

\begin{figure}[!htbp]
  \centering
  \includegraphics[width=0.98\textwidth,height=0.56\textheight,keepaspectratio]{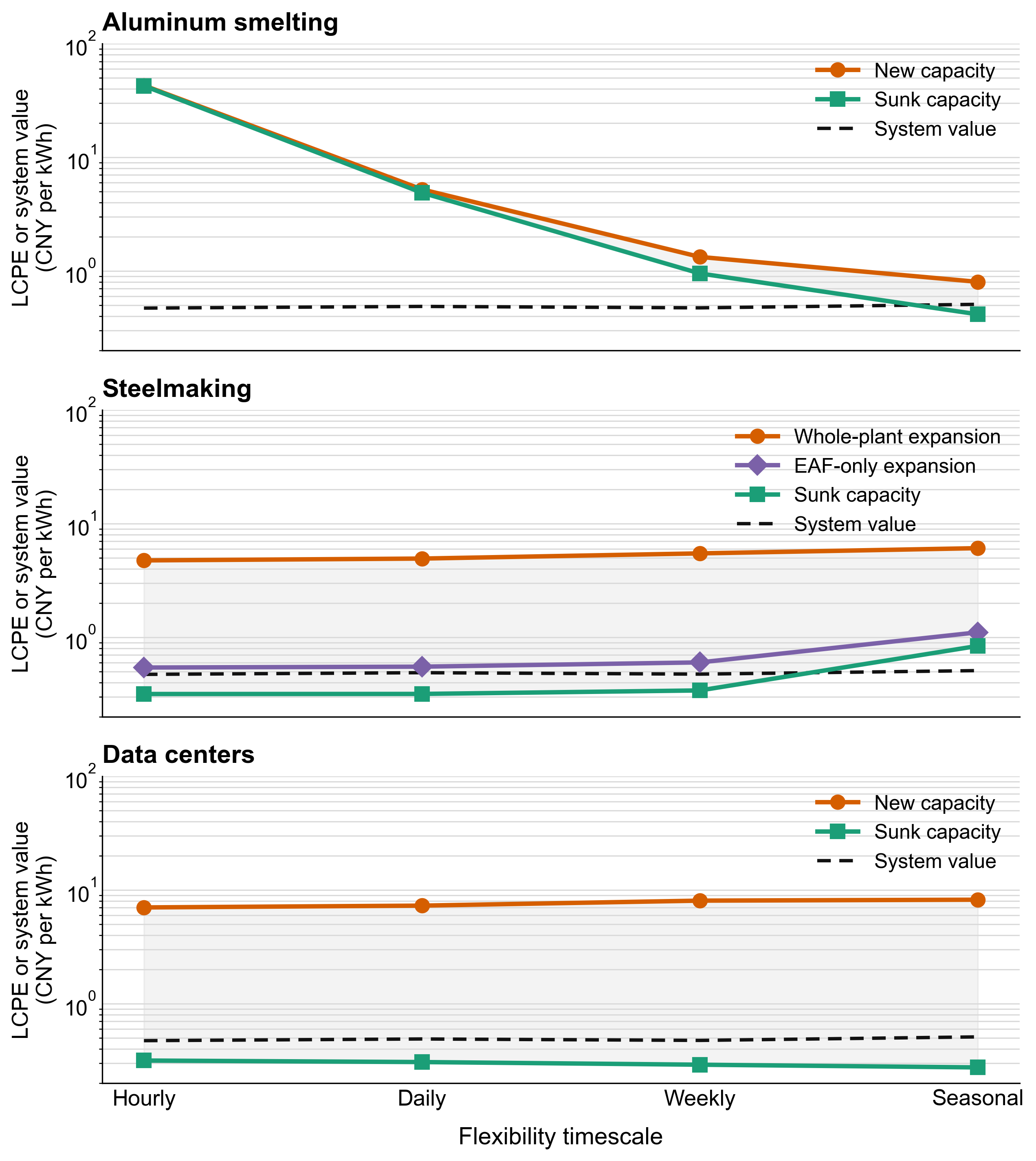}
  \caption{Demand-flexibility LCPE under newly added capacity, component-only expansion where available, and existing capacity treated as sunk. The dashed line is the Guangdong system value; gray shading spans the range across capacity boundaries within each sector. Spare existing capacity largely decides whether the load cases compared here are economic. Treating capacity as sunk, or expanding only a decoupled component such as the EAF, lowers steel and data-center LCPE by about an order of magnitude relative to whole-plant or shell/core-plus-fit-out expansion, and can bring several cases near or below the system-value line. At short duration aluminum is different. Restart cost dominates and removing capacity investment barely changes the hourly result, whereas sunk capacity matters more at seasonal scale. Data-center sunk-capacity values exclude workload-delay and service-quality costs.}
  \label{fig:capacity_boundaries}
\end{figure}

Whereas assessments of demand flexibility often assume that controllable capacity already exists, Fig.~\ref{fig:capacity_boundaries} shows that whether spare capacity is already available largely decides whether the load-side options compared here are economic. For steel, hourly LCPE falls from 4.75~CNY/kWh with whole-plant expansion to 0.542~CNY/kWh with EAF-only expansion and 0.317~CNY/kWh when existing EAF capacity is sunk. Data centers show a similar drop once shell/core and IT fit-out investment are removed, although the sunk result excludes delay and service-quality costs. With existing capacity treated as sunk, several load cases fall near or below the Guangdong system-value line; with new capacity required, they generally do not. Aluminum behaves differently at short duration. Restart cost dominates, so removing capacity investment changes hourly LCPE only from 43.0 to 42.6~CNY/kWh. At seasonal scale, sunk capacity lowers aluminum from 0.804 to 0.416~CNY/kWh and brings it below the modeled system-value line. Existing capacity can be treated as sunk only if its use for flexibility does not displace a profitable alternative use, accelerate replacement, or require additional supporting infrastructure. The comparison therefore does not imply that firms should build redundant production or computing capacity merely to obtain the lower sunk-capacity LCPE.

\subsection{Temporal scaling of cost components}

\begin{figure}[!htbp]
  \centering
  \includegraphics[width=0.98\textwidth,height=0.58\textheight,keepaspectratio]{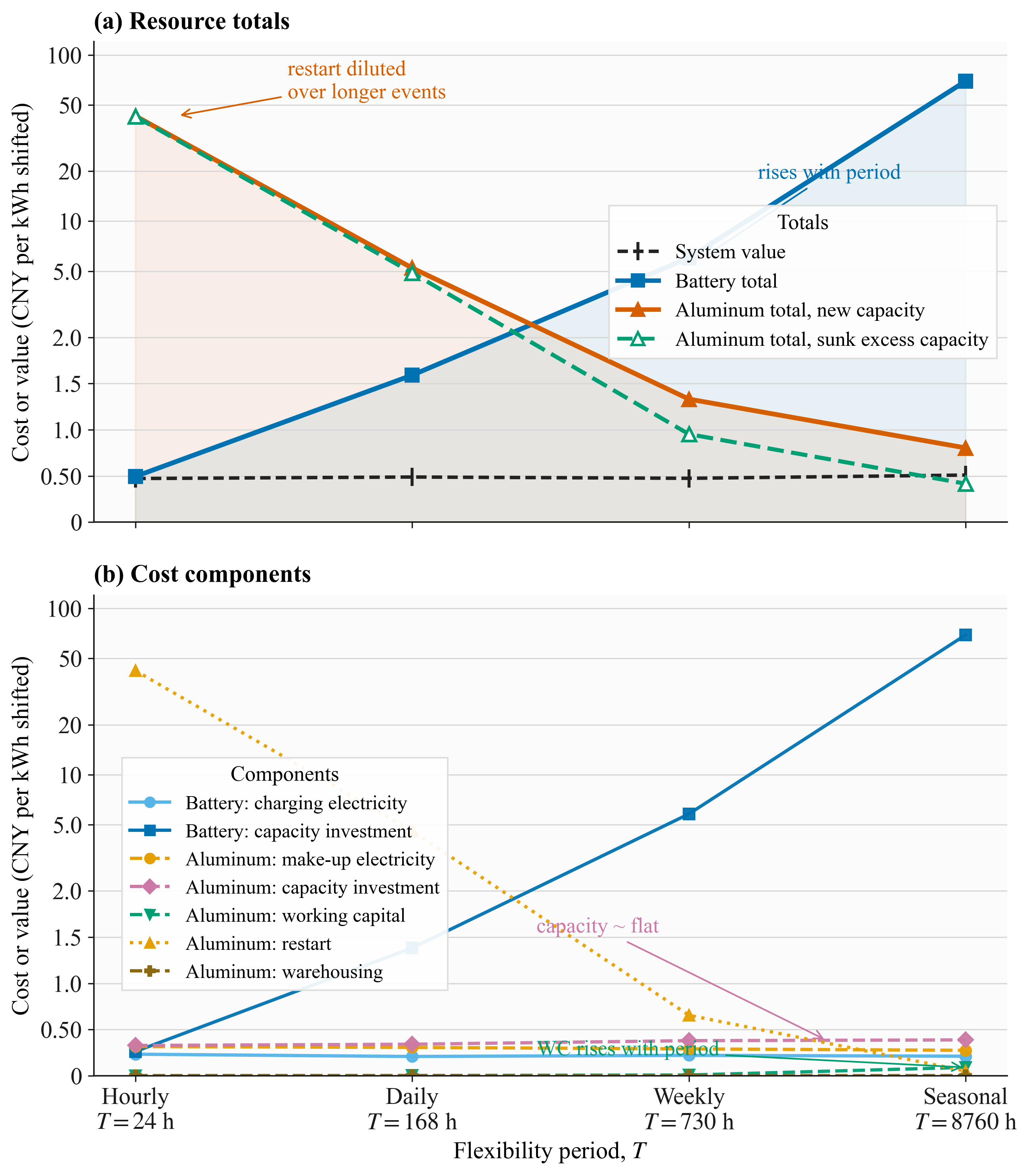}
  \caption{\textbf{(a)} Battery and aluminum-smelting LCPE versus flexibility period. Battery cost rises as annual cycle count falls, while aluminum totals fall as restart cost is diluted, producing the short- versus long-duration reversal in Fig.~\ref{fig:flex_costs}. \textbf{(b)} Cost components for aluminum. Under a fixed avoided-hour share, added-capacity cost is nearly independent of period length, working capital grows with the separation between production and sale, and event-based restart cost falls as the shifted duration within each event increases. Treating existing excess capacity as sunk removes the added-capacity term. Warehousing remains smaller than these drivers.}
  \label{fig:flex_scaling}
\end{figure}

As shown in Fig.~\ref{fig:flex_scaling}, the battery capacity term rises from 0.262~CNY/kWh in the 24-h case to 69.3~CNY/kWh when the same investment is recovered over a single annual cycle, while charging electricity changes little. Aluminum added-capacity cost stays nearly flat (0.327--0.388~CNY/kWh), consistent with the fixed avoided-hour-share derivation above. Working capital increases with the separation between production and sale, whereas restart cost falls from 42.3 to 0.051~CNY/kWh as the same event cost is spread over longer shifted durations. Short aluminum interruptions are therefore dominated by process disruption; seasonal aluminum shifting is divided mainly among make-up electricity, capacity, and inventory financing. Warehousing remains smaller than these terms. Longer duration does not automatically favor every load. Sunk-capacity steel still rises with working capital (0.317 to 0.841~CNY/kWh), and the apparently low data-center sunk-capacity path excludes delay and service-quality costs that would matter for seasonal computing delay.

\subsection{Limitations}\label{sec:limitations}

Communication, metering, market-access, participation, and retrofit costs are not assigned common values, nor are steel restart costs or data-center delay, SLA, and application-level network costs. Results for the affected cases are therefore lower bounds with respect to those categories: they mark a theoretical resource-side floor, not a delivered program cost. Where that floor already lies above the system-value reference, reducing enabling or transaction frictions cannot by itself make the option economic; the capacity, inventory, or disruption terms must first be brought down. The aluminum case models a full potline interruption rather than routine 5--10\% modulation that can remain continuous~\cite{taylor_low_amperage_2015,wong_studies_2023}.

Guangdong day-ahead prices are used because they align with the Chinese technology-cost inputs. In many other power markets, high--low price spreads are larger than in this sample, so which resource looks economic also depends on market design and local prices, not only on technology cost. The present study focuses on the technology side and therefore keeps system value as a simple high-price reference rather than a full market valuation. The same LCPE framework can be reapplied with other regional price series; only the value and charging or make-up price inputs need to change.

The numerical cases use representative central technology costs to show orders of magnitude and timescale dependence, not a universal crossover. Near the cost--value boundary, rankings remain project-specific. The battery--aluminum reversal, however, exceeds a factor of 80 in opposite directions at the short and seasonal scales; overturning that particular technology-side comparison would require correspondingly large changes in the input assumptions.

\section{Conclusion}\label{sec:conclusion}

Demand response is often treated as if it had no capital cost. Controllable load is assumed already available, so helping the power system balance supply and demand is taken to require only enabling measures or participation payments. From a life-cycle perspective that view is incomplete. Flexible operation means consuming less electricity in some hours and more in others while keeping the same product or service output. The make-up consumption requires additional productive or computing equipment whenever baseline assets already run near full utilization. The capital cost of that lower asset utilization, or of the spare capacity needed to restore it, is an implicit cost of demand flexibility and should be counted alongside make-up electricity, inventory financing, and process disruption.

In this article we formulate a levelized flexibility cost for large loads that can be compared with energy storage on a common per-shifted-kilowatt-hour basis. Storage and load-shifting costs are found to change in opposite directions with the flexibility period, so their relative economics reverse across timescales. Energy-storage capacity is used less often as the period lengthens, raising its annualized cost per unit of shifted energy. Under a fixed avoided-hour share, the added-capacity cost of load shifting is approximately independent of the period, whereas working capital increases with the time between production and sale, and an event-based disruption cost is spread over more shifted hours.

Comparing representative technology cases with electricity market prices, whether capacity investment is still required often decides whether demand flexibility is economic. Spare existing capacity can bring steel, data-center, and seasonal aluminum costs near or below the modeled system-value line; new whole-plant or shell/core-plus-fit-out capacity generally cannot. The capacity state matters less when process disruption dominates, as in the hourly full-interruption aluminum case. Where the modeled LCPE already exceeds system value, adding the enabling and transaction costs emphasized in much of the demand-response literature would only widen that gap.

Whether demand flexibility is still needed as batteries become cheaper has no simple answer. Neither resource is universally cheaper. Batteries beat demand flexibility at hourly scale when new production or computing capacity would otherwise be required, while demand flexibility can reverse that ranking at seasonal scale. In the battery--aluminum comparison with new smelting capacity included, aluminum LCPE is about 87 times battery LCPE in the hourly case, whereas battery LCPE is about 86 times aluminum LCPE in the seasonal case. Local electricity market spreads also affect which option clears a cost--value test, but the same levelized cost framework can be reapplied with other price series. The comparison therefore depends on timescale, asset state, and the full consumer-side cost of shifting.

\bibliographystyle{elsarticle-num}
\bibliography{reference}

\end{document}